\documentclass[11pt]{amsart}
\usepackage{setspace}
\usepackage{amsmath}
\usepackage{amstext}
\usepackage{amssymb}
\usepackage{amsthm} 
\usepackage{eurosym}
\usepackage{mathrsfs}
\usepackage{xcolor}
\usepackage{hyperref}
\usepackage{enumerate}   
\allowdisplaybreaks
\usepackage{geometry}
\newtheorem{thm}{Theorem}[section]

\newtheorem{lem}[thm]{Lemma}
\newtheorem{prop}[thm]{Proposition}
\theoremstyle{definition}
\newtheorem{defn}[thm]{Definition}
\newtheorem{ass}[thm]{Assumption}
\newtheorem{rem}[thm]{Remark}

\numberwithin{equation}{section}

\newcommand{\be}{\begin{equation}}
\newcommand{\ee}{\end{equation}}
\newcommand{\ba}{\begin{aligned}}
\newcommand{\ea}{\end{aligned}}
\newcommand{\cF}{\mathcal{F}}
\newcommand{\cB}{\mathcal{B}}
\newcommand{\bF}{\mathbb{F}}
\newcommand{\R}{\mathbb{R}}

\newcommand{\C}{\mathbb{C}}
\newcommand{\EE}{\mathbb{E}}
\newcommand{\cU}{\mathcal{U}}
\newcommand{\cD}{\mathcal{D}}
\newcommand{\cA}{\mathcal{A}}
\newcommand{\cI}{\mathcal{I}}
\newcommand{\cY}{\mathcal{Y}}

\newcommand{\im}{{\rm i}}
\newcommand{\real}{{\rm Re}}
\newcommand{\imag}{{\rm Im}}
\newcommand{\ud}{\mathrm{d}}
\newcommand{\PicF}{\pi^{c,\rm F}}
\newcommand{\PicB}{\pi^{c, \rm B}}
\newcommand{\PifF}{\pi^{f,\rm F}}
\newcommand{\PifB}{\pi^{f, \rm B}}
\newcommand{\PiF}{\Pi^{\rm F}}
\newcommand{\PiB}{\Pi^{\rm B}}

\title[Affine models for alternative risk-free rates]{Caplet pricing in affine models \\ for alternative risk-free rates}
\author[C. Fontana]{Claudio Fontana}
\address{Department of Mathematics ``Tullio Levi - Civita'', University of Padova, Italy.}
\email{fontana@math.unipd.it}
\date{\today}
\thanks{
Financial support from the Europlace Institute of Finance and the University of Padova (research programme BIRD190200/19) is gratefully acknowledged.}

\begin{document}

\maketitle

\begin{abstract}
Alternative risk-free rates (RFRs) play a central role in the reform of interest rate benchmarks. We study a model for RFRs driven by a general affine process. Under minimal assumptions, we derive explicit valuation formulas for forward-looking and backward-looking caplets/floorlets, term-basis caplets as well as 1-month and 3-month RFR futures contracts.
\end{abstract}

\section{Introduction}		\label{sec:intro}

The interest rate benchmarks reform is bringing a change of paradigm in fixed income markets. On 5 March 2021, the FCA announced that Libor rates will either cease to be provided or will no longer be representative benchmarks after 31 December 2021 (with the exception of US Libor rates for some tenors, that will be discontinued after June 2023)\footnote{See \url{https://www.fca.org.uk/news/press-releases/announcements-end-libor}.}.
The new benchmark rates, as well as fallback rates for existing contracts, are provided by alternative nearly {\em risk-free rates} (RFRs), which are determined by overnight rates backed by actual transactions. Such overnight rates include SOFR in the US, SONIA in the UK, \euro STR in the Euro area. 

The transition from Libor rates to RFRs has recently started to affect non-linear derivatives\footnote{As part of the {\em SOFR First} initiative, the Market Risk Advisory Committee (MRAC) of the Commodity Futures Trading Commission recommended to switch from Libor to SOFR in non-linear derivatives starting from 8 November 2021 (see \url{https://www.cftc.gov/PressRoom/PressReleases/8449-21}). A similar recommendation has been issued by the working group on Sterling risk-free reference rates starting on 11 May 2021.}.
The trading volume in SOFR caps/floors reached 926.9 USD bn in the first 9 months of 2022, increasing from 85.6 USD bn in the whole year 2021 (source: ISDA). In the case of SONIA caps/floors, the trading volume in the first 9 months of 2022 amounts to 210.9 USD bn, against 72.5 USD bn in the whole year 2021 (source: ISDA).

In the post-Libor universe, one can consider {\em forward-looking} and {\em backward-looking} caps/floors, depending on the rate which defines their payoff. 
While forward-looking caps/floors are based on forward-looking term rates and are conceptually similar to classical Libor caps/floors, backward-looking caps/floors have a significantly different nature, since their payoff is determined by the compounded in-arrears RFR (see \cite{LM19,Pit20}). 
While both types of caps/floors coexist in the current post-Libor market, backward-looking caps/floors play a particularly important role. Indeed, forward-looking term rates (such as CME term SOFR and ICE term SOFR) have been introduced only recently and are not supported by the Alternative Reference Rates Committee (ARRC) for use in derivatives markets, being restricted to derivatives that hedge cash products referencing term SOFR.
In addition, the Libor fallbacks protocol adopts backward-looking rates as fallback rates in existing contracts (see \cite{ISDA}).

In this paper, we derive pricing formulas for forward-looking and backward-looking options in the context of a short-rate RFR model driven by a general affine process. While the valuation of forward-looking payoffs is relatively straightforward, the backward-looking case requires a more elaborate analysis and has never been considered in the literature on affine processes, except for the specific case of Gaussian Hull-White models (see below for more details).
By relying on Fourier methods and a study of the integrability properties of certain functionals of the affine process, we obtain pricing formulas expressed as one-dimensional integrals, which can be efficiently implemented with fast Fourier transform (see \cite{cm99}).
We work under minimal technical assumptions, without imposing ad-hoc integrability requirements.
We also study {\em term-basis} caplets, corresponding to exchange options between forward-looking and backward-looking rates (see \cite{LM19}).
Moreover, we derive alternative pricing formulas based on forward measures, which can be explicitly computed for some popular interest rate models, including multi-factor Hull-White and Wishart models.
Finally, we study futures contracts, currently the most liquid RFR-based products. In the case of 1-month RFR futures, we obtain a pricing formula that does not even require the explicit characteristic function of the driving process.

Short-rate modeling is probably the most natural approach for modeling RFRs. In \cite{Merc18}, one of the first papers on SOFR modeling, a Gaussian Hull-White short-rate model is adopted. 
In recent short-rate approaches to RFR modeling, the Hull-White model has remained dominant. 
This is for instance the case of \cite{Hasegawa21,Hof20,Turf21,Xu22}, where pricing formulae for backward-looking caplets are derived\footnote{We want to point out that, in the specific context of a Gaussian one-factor Hull-White model, a backward-looking caplet can be priced by relying on the valuation formulas obtained in the earlier work \cite{Hen04}.}, and also of \cite{RB21}, where in addition collateralization and funding costs are taken into account.
However, the Hull-White model does not support volatility smiles (see, e.g., \cite{Pit20}). Moreover, as shown in \cite{AB20}, spikes and jumps are a prominent feature of RFRs and have a sizable effect on the pricing of backward-looking caplets.
This motivates the modeling of RFRs by means of general affine processes and the study of their pricing aspects. 
Let us also mention that backward-looking caplets have been analyzed in the seminal work \cite{LM19} in the context of an extended Libor market model, in \cite{Will21} for the SABR model and in \cite{MS20} by adopting a rational model for the savings account associated to an RFR.
A short-rate approach for modelling RFRs based on  affine semimartingales has been developed in \cite{FGrS22}, also allowing for jumps at predetermined times, but without a specific focus on pricing applications.

The paper is structured as follows. In Section \ref{sec:affine}, we recall some essential notions on affine processes and prove two additional  properties of the solutions to the associated Riccati ODEs. Section \ref{sec:model} contains the description of the modeling framework and the definition of forward/backward-looking rates. In Section \ref{sec:caplets}, we derive pricing formulas for backward/forward-looking caplets/ floorlets and term-basis caplets, while in Section \ref{sec:example}, we derive pricing formulas for forward-looking and backward-looking caplets in the context of a CIR++ model.
In Section \ref{sec:futures} we consider the valuation of futures contracts.

\section{Preliminaries on affine processes} 	\label{sec:affine}

In this section, we present some general results on affine processes that will be needed in the sequel. 
We let $(\Omega,\cF,P)$ be a probability space endowed with a filtration $\bF=(\cF_t)_{t\geq0}$, satisfying the usual conditions of right-continuity and completeness, and $X=(X_t)_{t\geq0}$ be a c\`adl\`ag adapted time-homogeneous conservative Markov process on $(\Omega,\cF,\bF,P)$, taking values in the state space $D:=\R^m_+\times\R^n$.\footnote{We restrict our attention to affine processes on $\R^m_+\times\R^n$ for simplicity of presentation. The modeling framework developed in this paper can be readily extended to matrix-valued affine processes, as characterized in \cite{CFMT11}, and all main results remain valid in the matrix-valued case with identical statements. In particular, this is possible since the results of \cite{krm12} are also applicable to affine processes taking values in the cone of symmetric positive-semidefinite matrices, up to an adaptation of the notation and of the proofs of Lemmata \ref{lem:Y_t} and \ref{lem:convex}.}
The family of the transition kernels of the Markov process $X$ is given by $\{p_t : D\times\cB_D\to[0,1];t\geq0\}$, where $\cB_D$ denotes the Borel $\sigma$-algebra of $D$.
We also introduce the set $\cU:=\C^m_-\times\im\R^n$, where $\C_-:=\{u\in\C : \real(u)\leq0\}$.
Setting $d:=m+n$, we recall the following definition (see \cite{dfs03} and \cite[Definition 2.2]{krm12}).

\begin{defn}	\label{def:affine}
The process $X$ is called {\em affine} with state space $D$ if
\begin{enumerate}[(i)]
\item it is stochastically continuous (i.e., its transition kernels satisfy $\lim_{s\rightarrow t}p_s(x,\cdot)=p_t(x,\cdot)$ weakly for all $(t,x)\in\R_+\times D$) and
\item there exist functions $\phi:\R_+\times\cU\rightarrow\C$ and $\psi:\R_+\times\cU\rightarrow\C^d$ such that
\be	\label{eq:Fourier_affine}
\int_De^{\langle u,\xi\rangle}p_t(x,\ud\xi)
= e^{\phi(t,u)+\langle\psi(t,u),x\rangle},
\ee
for all $(t,u)\in\R_+\times\cU$ and for every $x\in D$.
\end{enumerate}
\end{defn}

Requirement (i) in Definition \ref{def:affine} implies the regularity of the affine process $X$ (see  \cite[Theorem 5.1]{KRST11}), meaning that the derivatives $\partial_t\phi(t,u)|_{t=0}$ and $\partial_t\psi(t,u)|_{t=0}$ exist for all $u\in\cU$ and are continuous at $u=0$. As a consequence, in view of \cite[Theorem 2.7]{dfs03}, the functions $\phi$ and $\psi$ in \eqref{eq:Fourier_affine} are determined by the following system of generalized Riccati ODEs:
\[	\begin{aligned}
\frac{\partial \phi(t,u)}{\partial t}&=F\bigl(\psi(t,u)\bigr),\qquad \phi(0,u)=0,\\
\frac{\partial \psi(t,u)}{\partial t}&=R\bigl(\psi(t,u)\bigr),\qquad \psi(0,u)=u\in\cU,
\end{aligned}	\]
with the functions $F$ and $R$ admitting explicit representations of L\'evy-Khintchine type as follows:
\be	\label{eq:LevyKhint}\ba
F(u) &= \langle \alpha_0u,u\rangle + \langle \beta_0,u\rangle + \int_{D\setminus\{0\}}\bigl(e^{\langle u,\xi\rangle}-1-\langle u,h(\xi)\rangle\bigr)\mu_0(\ud\xi),	\\
R_i(u) &:= \langle \alpha_iu,u\rangle + \langle \beta_i,u\rangle + \int_{D\setminus\{0\}}\bigl(e^{\langle u,\xi\rangle}-1-\langle u,h(\xi)\rangle\bigr)\mu_i(\ud\xi),	
\quad\text{ for all }i=1,\ldots,d,\\
\ea	\ee
with respect to a set of parameters $(\alpha,\beta,\mu)$ satisfying the admissibility requirements of \cite[Definition 2.6]{dfs03}, where $\mu_i$, $i=0,1,\ldots,d$, are L\'evy measures on $D$ and $h$ is a truncation function. 
The functions $F$ and $R$ completely characterize the law of the affine process $X$ and are therefore called the functional characteristics of $X$.

In the following, we shall work with time integrals of the affine process $X$. More specifically, let $\Lambda\in\R^d$ and consider the process $Y=(Y_t)_{t\geq0}$ defined by $Y_t:=\int_0^t\langle\Lambda,X_s\rangle\ud s$, for all $t\geq0$. Note that, since $X$ is c\`adl\`ag, the integral is well-defined pathwise. The couple $(X,Y)$ can be regarded as a process on the enlarged state space $D\times\R$. The following well-known result, which is a direct consequence of \cite[Proposition 11.2]{dfs03}, asserts that $(X,Y)$ is an affine process.

\begin{prop}	\label{prop:extended_aff}
Let the process $(X,Y)$ be defined as above. Then, $(X,Y)$ is an affine process on the state space $D\times\R$ and, for all $(t,u,v)\in\R_+\times\cU\times\im\R$, it holds that
\be	\label{eq:conditional_Fourier}
\EE\bigl[e^{\langle u,X_T\rangle+vY_T}\big|\cF_t\bigr]
= e^{\Phi(T-t,u,v)+\langle\Psi(T-t,u,v),X_t\rangle+vY_t},
\ee
where $\Phi:\R_+\times\cU\times\im\R\rightarrow\C$ and $\Psi:\R_+\times\cU\times\im\R\rightarrow\C^d$ are solutions to
\begin{subequations}
\begin{align}	
\frac{\partial \Phi(t,u,v)}{\partial t}&=F\bigl(\Psi(t,u,v)\bigr),\qquad\qquad \;\,\Phi(0,u,v)=0,
\label{eq:Riccati1}\\
\frac{\partial \Psi(t,u,v)}{\partial t}&=R\bigl(\Psi(t,u,v)\bigr)+v\Lambda,\qquad \Psi(0,u,v)=u.
\label{eq:Riccati2}
\end{align}
\end{subequations}
\end{prop}

\begin{rem}	\label{rem:structure_psi}
For a generic vector $x\in\C^d$, let us introduce the notation $x=(x_I,x_J)\in\C^m\times\C^n$. Writing $\Psi(t,u,v)=(\Psi_I(t,u,v),\Psi_J(t,u,v))$, it can be seen that $\Psi_J(t,u,v)$ solves a linear ODE with initial value $u_J$ and is therefore globally well-defined on $\R_+$  (see \cite[Section 11.2]{dfs03}). In particular, $\Psi_J(t,u,v)$ is linear in $(u_J,v)$ and does not depend on $u_I$.
Note also that the function $\Phi(\cdot,u,v)$ can be explicitly solved as $\Phi(t,u,v)=\int_0^tF(\Psi(s,u,v))\ud s$, for all $(t,u,v)\in\R_+\times\cU\times\im\R$. 
\end{rem}

For pricing applications, we need to extend the domain of the affine transform formula \eqref{eq:conditional_Fourier} of the joint process $(X,Y)$ beyond the set $\cU\times\im\R$. This issue has been studied in detail in \cite{krm12}, whose approach is followed here. 
As a preliminary, noting that the admissibility requirements of \cite[Definition 2.6]{dfs03} imply that $\mu_j=0$, for all $j=m+1,\ldots,n$, let us define the set
\[
\cY := \biggl\{y\in\R^d : \sum_{i=0}^m\int_{\{\xi\in D : |\xi|\geq1\}}e^{\langle y,\xi\rangle}\mu_i(\ud\xi)<+\infty\biggr\}.
\]
The set $\cY$ is non-empty and convex and constitutes the effective real domain of the functions $F$ and $R$ given in \eqref{eq:LevyKhint}, which are therefore well-defined convex functions on $\cY$. 
Note that, by definition, the set $\cY$ only depends on the behavior of the large jumps of $X$. 

\begin{ass}	\label{ass:0}
It holds that $0\in\cY^{\circ}$, with $\cY^{\circ}$ denoting the interior of $\cY$.
\end{ass}

Noting that $0\in\cY$ always holds, Assumption \ref{ass:0} represents a mild technical requirement that can be easily checked, since the measures $\mu_i$, $i=0,1,\ldots,m$, are explicitly known in applications. In particular, Assumption \ref{ass:0} is trivially satisfied by every continuous affine process, since in that case $\cY=\R^d$.
Under Assumption \ref{ass:0}, the set $S(\cY^{\circ}):=\{u\in\C^d : \real(u)\in\cY^{\circ}\}$ is non-empty and the functions $F$ and $R$ can be analytically extended to $S(\cY^{\circ})$. This enables us to study the Riccati ODEs \eqref{eq:Riccati1}-\eqref{eq:Riccati2} for $(u,v)\in S(\cY^{\circ})\times\C$, replacing $F$ and $R$ by their analytic extensions to $S(\cY^{\circ})$.
Since $F$ and $R$ are locally Lipschitz on $S(\cY^{\circ})$, the (possibly local) solution $(\Phi,\Psi)$ to \eqref{eq:Riccati1}-\eqref{eq:Riccati2} is unique if constrained to stay in the open domain $S(\cY^{\circ})$.
For $(u,v)\in S(\cY^{\circ})\times\C$, let us denote by $T_+(u,v)$ the maximal lifetime of the solution to the  Riccati system \eqref{eq:Riccati1}-\eqref{eq:Riccati2} such that $\Psi(t,\real(u),\real(v))\in\cY^{\circ}$ for all $t\leq T_+(u,v)$.  

The following proposition is a direct consequence of \cite[Theorems 2.14 and 2.26]{krm12} and provides the desired extension of the affine transform formula \eqref{eq:conditional_Fourier}.

\begin{prop}	\label{prop:ext_Fourier}
Suppose that Assumption \ref{ass:0} holds. Then, the following hold:
\begin{enumerate}[(i)]
\item for $(u,v)\in\cY^{\circ}\times\R$, it holds that $\EE[e^{\langle u,X_T\rangle+vY_T}]<+\infty$ for all $T\leq T_+(u,v)$;
\item for $(u,v)\in S(\cY^{\circ})\times\C$, it holds that
\be	\label{eq:ext_Fourier}
\EE\bigl[e^{\langle u,X_T\rangle+vY_T}\big|\cF_t\bigr]
= e^{\Phi(T-t,u,v)+\langle\Psi(T-t,u,v),X_t\rangle+vY_t},
\ee
for all $0\leq t\leq T<+\infty$ such that $T-t\leq T_+(u,v)$, where $\Phi(\cdot,u,v)$ and $\Psi(\cdot,u,v)$ solve \eqref{eq:Riccati1}-\eqref{eq:Riccati2} with $F$ and $R$ analytically extended to $S(\cY^{\circ})$.
\end{enumerate}
\end{prop}

For convenience of notation, similarly as in \cite{HKRS17}, let us introduce the sets 
\[
\cY_t := \bigl\{(u,v)\in\cY^{\circ}\times\R : T_+(u,v)>t\bigr\},
\qquad
\cD_t := S(\cY_t) = \bigl\{(u,v)\in\C^{d+1}:(\real(u),\real(v))\in\cY_t\bigr\},
\]
for $t\geq0$. The sets $\cY_t$ and $\cD_t$ are open and the results of \cite{KRMS10} imply that $\cY_t$ is convex. Moreover, the affine transform formula \eqref{eq:ext_Fourier} holds for all $(u,v)\in\cD_{T-t}$.

In the following, we denote by $\preceq$ the partial order on $\R^m_+$ (i.e., for any $y,u\in \R^m_+$, the relation $y\preceq u$ means that $u-y\in\R^m_+$) and recall the notation $x=(x_I,x_J)$ for any vector $x\in D=\R^m_+\times\R^n$. It is easily seen that the set $\cY$ is order-preserving. In the following lemma, we show that this property extends to the set $\cY_t$, for all $t\geq0$.

\begin{lem}	\label{lem:Y_t}
Let $(u,v)\in\cY_t$, for some $t\geq0$. If $(y,z)\in D\times\R$ is such that $y_I\preceq u_I$, $y_J=u_J$ and $z=v$, then $(y,z)\in\cY_t$.
\end{lem}
\begin{proof}
By \cite[Lemma 5.8]{krm12}, it holds that $y\in\cY^{\circ}$. Since $X_t$ takes values in $D$, it holds that
\[
\EE\bigl[e^{\langle y,X_t\rangle+zY_t}\bigr]
= \EE\bigl[e^{\langle y_I,X_{I,t}\rangle+\langle y_J,X_{J,t}\rangle+zY_t}\bigr]
\leq \EE\bigl[e^{\langle u_I,X_{I,t}\rangle+\langle u_J,X_{J,t}\rangle+vY_t}\bigr]
= \EE\bigl[e^{\langle u,X_t\rangle+vY_t}\bigr]<+\infty,
\]
due to the assumption that $(u,v)\in\cY_t$ together with part (i) of Proposition \ref{prop:ext_Fourier}. 
By \cite[Theorem 2.14-(a)]{krm12}, this implies the existence of a solution $(\Phi(\cdot,y,z),\Psi(\cdot,y,z))$ to \eqref{eq:Riccati1}-\eqref{eq:Riccati2} up to time $t$. It remains to prove that $\Psi(s,y,z)\in\cY^{\circ}$ for all $s\leq t$. To this effect, note first that $\Psi_J(s,u,v)=\Psi_J(s,y,z)$, for all $s\geq0$, since $\Psi_J(\cdot,u,v)$ does not depend on $u_I$ (see Remark \ref{rem:structure_psi}) and $(u_J,v)=(y_J,z)$. 
Hence, considering the first $m$ components of \eqref{eq:Riccati2}, we have that
\[
\frac{\partial}{\partial s}\Psi_I(s,y,z) = R_I\bigl(\Psi_I(s,y,z),\Psi_J(s,u,v)\bigr)+v\Lambda_I,
\qquad\Psi_I(0,y,z)=y_I.
\]
By \cite[Lemma 5.7]{krm12}, the map $x_I\to R_I(x_I,\Psi_J(s,u,v))$ is quasi-monotone increasing with respect to the natural cone $\R^m_+$, for all $s\geq0$. The comparison result of \cite{Volk73} implies that $\Psi_I(s,y,z)\preceq\Psi_I(s,u,v)$, for all $s<T_+(y,z)\wedge T_+(u,v)$. Arguing by contradiction, suppose that $T_+(y,z)<T_+(u,v)$. Since we already know that $\Psi(\cdot,y,z)$ cannot explode before time $t$, this means that $\Psi(\cdot,y,z)$ reaches the boundary of $\cY$ before time $T_+(u,v)$. However, continuity implies that $\Psi_I(T_+(y,z),y,z)\preceq\Psi_I(T_+(y,z),u,v)$ and, since $T_+(y,z)<T_+(u,v)$, we have that $\Psi(T_+(y,z),u,v)\in\cY^{\circ}$. 
Again by \cite[Lemma 5.8]{krm12}, this implies that $\Psi(T_+(y,z),y,z)\in\cY^{\circ}$, thus obtaining a contradiction. 
Therefore, $T_+(y,z)\geq T_+(u,v)$ must necessarily hold. Since $T_+(u,v)>t$ by assumption, this shows that $(y,z)\in\cY_t$.
\end{proof}

We close this section with the following result on the convexity of the function $\Psi_I$. While it can be deduced from the results of \cite{KRMS10}, we prefer to provide a self-contained direct proof.

\begin{lem}	\label{lem:convex}
Let $(u_1,v_1),(u_2,v_2)\in\cY_t$, for some $t\geq0$, and define $(u_{\lambda},v_{\lambda}):=\lambda(u_1,v_1)+(1-\lambda)(u_2,v_2)$, for $\lambda\in[0,1]$. Then it holds that $\Psi_I(s,u_{\lambda},v_{\lambda}) \preceq \lambda\Psi_I(s,u_1,v_1) + (1-\lambda)\Psi_I(s,u_2,v_2)$, for all $s\in[0,t]$.
\end{lem}
\begin{proof}
For brevity of notation, let us define $f(s):=\lambda\Psi_I(s,u_1,v_1)+(1-\lambda)\Psi_I(s,u_2,v_2)$, for all $s\in[0,t]$.
Since the map $(u,v)\mapsto\Psi_J(s,u,v)$ is linear and does not depend on $u_I$ (see Remark \ref{rem:structure_psi}), we have that $\Psi_J(s,u_{\lambda},v_{\lambda}) = \lambda\Psi_J(s,u_1,v_1) + (1-\lambda)\Psi_J(s,u_2,v_2)$. 
By \eqref{eq:Riccati2} together with the convexity of $R_I$, it holds that
\begin{align*}
&\frac{\partial}{\partial s}\Psi_I(s,u_{\lambda},v_{\lambda})
- R_I\bigl(\Psi_I(s,u_{\lambda},v_{\lambda}),\Psi_J(s,u_{\lambda},v_{\lambda})\bigr)	= \Lambda_I v_{\lambda}	
= \lambda\Lambda_I v_1 + (1-\lambda)\Lambda_I v_2	\\
&= \lambda\left(\frac{\partial}{\partial s}\Psi_I(s,u_1,v_1)
- R_I\bigl(\Psi_I(s,u_1,v_1),\Psi_J(s,u_1,v_1)\bigr)\right)	\\
&\quad
+(1-\lambda)\left(\frac{\partial}{\partial s}\Psi_I(s,u_2,v_2)
- R_I\bigl(\Psi_I(s,u_2,v_2),\Psi_J(s,u_2,v_2)\bigr)\right)	\\
&\leq \frac{\partial}{\partial s}f(s) - R_I\bigl(f(s),\Psi_J(s,u_{\lambda},v_{\lambda})\bigr).
\end{align*}
By assumption, $T_+(u_1,v_1)\wedge T_+(u_2,v_2)>t$. Moreover, due to the convexity of $\cY_t$, we have that $T_+(u_{\lambda},v_{\lambda})>t$. This implies that $(f(s),\Psi_J(s,u_{\lambda},v_{\lambda}))\in\cY^{\circ}$ and $(\Psi_I(s,u_{\lambda},v_{\lambda}),\Psi_J(s,u_{\lambda},v_{\lambda}))\in\cY^{\circ}$, for all $s\in[0,t]$.
The claim then follows from the comparison result of \cite{Volk73}, using the fact that the map $x_I\mapsto R_I(x_I,\Psi_J(s,u_{\lambda},v_{\lambda}))$ is quasi-monotone increasing, for all $s\in[0,t]$.
\end{proof}

\begin{rem}	\label{rem:semiflow}
For later use, we recall that the functions $\Phi(\cdot,u,v)$ and $\Psi(\cdot,u,v)$ appearing in Proposition \ref{prop:ext_Fourier} satisfy the following semiflow relations, for all $(u,v)\in\cY^{\circ}\times\R$:
\be	\label{eq:semiflow}\begin{aligned}
\Phi(s+t,u,v)&=\Phi(t,u,v)+\Phi(s,\Psi(t,u,v),v),	\\
\Psi(s+t,u,v)&=\Psi(s,\Psi(t,u,v),v),
\end{aligned}	\ee
for all $0\leq s\leq t<+\infty$ such that $s+t\leq T_+(u,v)$ (see \cite[Lemma 4.3]{krm12}).
\end{rem}

\section{Model setup}	\label{sec:model}

In this section, we present the general affine modeling framework  for a risk-free overnight rate. Let $X=(X_t)_{t\geq0}$ be an affine process on $(\Omega,\cF,\bF,P)$ taking values in the space $D=\R^m_+\times\R^n$. 
Adopting a short rate approach, we model the instantaneous RFR process $r=(r_t)_{t\geq0}$ as follows:
\be	\label{eq:short_rate}
r_t := \ell(t) + \langle\Lambda,X_t\rangle,
\qquad\text{ for all }t\geq0,
\ee
where $\Lambda\in\R^d$, with $d=m+n$, and $\ell:\R_+\rightarrow\R$ is a function satisfying $\int_0^T|\ell(t)|\ud t<+\infty$, for all $T>0$.
In line with \cite{BM01}, the function $\ell$ in \eqref{eq:short_rate} serves to fit the RFR term structure at $t=0$.
We assume that $r$ also represents the collateral rate for collateralized OTC transactions, as well as the price alignment interest (PAI) for cleared derivatives. As pointed out by \cite{RB21}, this assumption is consistent with the collateralization schemes currently prevailing in the market.
Adopting a martingale approach, we can therefore assume that $P$ is a pricing measure with respect to the savings account $B=(B_t)_{t\geq0}$ given by $B_t:=\exp(\int_0^tr_s\ud s)$, $t\geq0$.
Denoting by $P_t(T)$ the price at time $t$ of a zero-coupon bond (ZCB) with maturity $T$, it holds that
\be	\label{eq:bond_price}
P_t(T) = \EE[B_t/B_T|\cF_t],
\qquad\text{ for all }0\leq t\leq T<+\infty.
\ee

We denote $Y:=\int_0^{\cdot}\langle\Lambda,X_s\rangle\ud s$, similarly as in Section \ref{sec:affine}, and introduce the following  assumption.

\begin{ass}	\label{ass:bond}
For every $T>0$, it holds that $(0,-1)\in\cY_T$.
\end{ass}

In view of Proposition \ref{prop:ext_Fourier}, Assumption \ref{ass:bond} ensures that ZCB prices are well-defined by \eqref{eq:bond_price} for all maturities $T>0$. Note also that Assumption \ref{ass:bond} implies the validity of Assumption \ref{ass:0}. 
In the following, we always suppose that Assumption \ref{ass:bond} is satisfied without further mention.
We emphasize that the pricing formulae obtained in Section \ref{sec:caplets} will not rely on any additional technical or integrability requirement beyond Assumption \ref{ass:bond}.

In line with the terminology of \cite{LM19}, we give the following definition, denoting by $\EE^T$ the expectation under the $T$-forward measure $P^T$ defined by $\ud P^T/\ud P:=1/(B_TP_0(T))$, for $T>0$.

\begin{defn}	\label{def:rates}
For all $0\leq S<T<+\infty$, we define: 
\begin{enumerate}[(i)]
\item the {\em backward-looking rate} $R(S,T):=(B_T/B_S-1)/(T-S)$;
\item the {\em forward-looking rate} $F(S,T):=\EE^T[R(S,T)|\cF_S]$. 
\end{enumerate}
\end{defn}

We point out that part (ii) of Definition \ref{def:rates} is well-posed since ZCB prices are well-defined for all maturities. More specifically, for all $0\leq S<T<+\infty$, it holds that
\be	\label{eq:fwd_rate}
1+(T-S)F(S,T) 
= \EE^T[B_T/B_S|\cF_S]
= 1/P_S(T).
\ee
The fundamental difference between forward-looking and backward-looking rates consists in the fact that the forward-looking rate $F(S,T)$ is observable at time $S$, while the backward-looking rate $R(S,T)$ is only known at the end of the accrual period $[S,T]$. According to Definition \ref{def:rates}, the forward-looking rate $F(S,T)$ represents the fixed rate $K$ that makes equal to zero the value at time $S$ of a swap (overnight indexed swap) delivering payoff $R(S,T)-K$ at maturity $T$.

For brevity of notation, we introduce the following shorthand notation, for $(u,v)\in S(\cY^{\circ})\times\C$:
\begin{align*}
A^0(t,t+\tau,v) &:= \Phi(\tau,0,-v)-vL(t,t+\tau),
&B^0(\tau,v) &:= \Psi(\tau,0,-v),\quad&\text{for all }\tau\leq T_+(0,-v),	\\
A^1(t,t+\tau,u) &:= \Phi(\tau,u,-1)-L(t,t+\tau),
&B^1(\tau,u) &:= \Psi(\tau,u,-1),\quad&\text{for all }\tau\leq T_+(u,-1),
\end{align*}
where $L(t,t+\tau):=\int_t^{t+\tau}\!\ell(z)\ud z$ and $(\Phi,\Psi)$ solves \eqref{eq:Riccati1}-\eqref{eq:Riccati2}. Making use of this notation and in view of Proposition \ref{eq:ext_Fourier} and Assumption \ref{ass:bond}, ZCB prices can be expressed as follows:
\be	\label{eq:ZCB_affine}
P_t(T) = e^{A^0(t,T,1)+\langle B^0(T-t,1),X_t\rangle},
\qquad\text{ for all }0\leq t\leq T<+\infty.
\ee

\begin{rem}	\label{rem:approx}
In reality, the backward-looking rate is computed by compounding at a daily frequency the overnight rate over the period $[S,T]$, corresponding to the following quantity:
\be	\label{eq:real_rate}
R'(S,T) := \frac{1}{T-S}\left(\prod_{i=1}^n\frac{1}{P_{t_i}(t_i+\delta_i)}-1\right),
\ee
where the product is taken over the business days $(t_1,\ldots,t_n)$ comprised between $S$ and $T$, with $\delta_i$ denoting the day-count fraction, for $i=1,\ldots,n$. Due to its affine structure, our model allows for an explicit representation of the discretely compounded rate $R'(S,T)$, given by
\be	\label{eq:formula_real_rate}
R'(S,T) = \frac{1}{T-S}\Bigl(e^{-\sum_{i=1}^nA^0(t_i,t_i+\delta_i,1)-\sum_{i=1}^n\langle B^0(\delta_i,1),X_{t_i}\rangle}-1\Bigr).
\ee
The specification of the backward-looking rate $R(S,T)$ adopted in  Definition \ref{def:rates} corresponds to the usual continuous-time approximation of \eqref{eq:real_rate} (see, e.g., \cite{LM19}) and is adopted for simplicity of computation only. The approximation error can be quantified exactly by relying on formula \eqref{eq:formula_real_rate} and can be shown to be practically negligible (see, e.g., \cite[Appendix D]{SS21}).
\end{rem}

\begin{rem}
The modeling framework can be extended to time-inhomogeneous affine processes, with time-dependent parameters $(\alpha,\beta,\mu)$ in \eqref{eq:LevyKhint}. All pricing results derived in the next section can be easily generalized to this case, provided that suitable exponential moments exists.
For time-inhomogeneous affine processes, under an additional integrability condition on the jump measure, it has been shown in \cite[Theorem 5.1]{KMK10} that the transform formula \eqref{eq:ext_Fourier} holds for real $(u,v)$ if there exists a solution to the Riccati equations \eqref{eq:Riccati1}-\eqref{eq:Riccati2} up to time $T$. 
However, while in the case of time-homogeneous affine processes we are able to characterize and provide conditions for the existence of solutions to the Riccati equations up to a certain time, for time-inhomogeneous affine processes one has to assume it a priori, similarly as in \cite[Section 5]{KMK10}. For this reason, in the present paper we restrict our attention to time-homogeneous affine processes.
We mention that the extension to time-inhomogeneous processes can be of practical interest in view of calibrating the model simultaneously to different products, such as forward/backward-looking caplets and term-basis caplets (see Section \ref{sec:caplets}).
\end{rem}

\section{Pricing of options on forward and backward-looking rates}
\label{sec:caplets}

As explained in Section \ref{sec:intro}, one can consider two distinct types of caplets/floorlets, depending on whether the payoff is determined by the forward-looking or the backward-looking rate. 

\begin{defn}	\label{def:caplets}
For $0\leq S<T<+\infty$ and $K>0$,
\begin{enumerate}[(i)]
\item a {\em forward-looking caplet} is defined by the payoff $(T-S)(F(S,T)-K)^+$ at maturity $T$;
\item a {\em backward-looking caplet} is defined by the payoff $(T-S)(R(S,T)-K)^+$ at maturity $T$.
\end{enumerate}
\end{defn} 

Forward-looking and backward-looking floorlets are defined in an analogous way. Since the reference probability $P$ is assumed to be a pricing measure with respect to the num\'eraire $B$, arbitrage-free prices of forward-looking and backward-looking caplets can be expressed as follows:
\be	\label{eq:general_formulae}\begin{aligned}
\PicF_t(S,T,K) 
&:= (T-S)\,\EE\left[\frac{B_t}{B_T}(F(S,T)-K)^+\Big|\cF_t\right],	\\
\PicB_t(S,T,K) 
&:= (T-S)\,\EE\left[\frac{B_t}{B_T}(R(S,T)-K)^+\Big|\cF_t\right].
\end{aligned}\ee
The prices of forward-looking and backward-looking floorlets can be expressed in a similar way and shall be denoted by $\PifF_t(S,T,K)$ and $\PifB_t(S,T,K)$, respectively.

\begin{rem}
(1) As pointed out by several authors (see \cite{LM19,MS20,Pit20}), Definition \ref{def:rates} together with Jensen's inequality implies the following relation between the two types of caplets:
\begin{align*}
\PicF_t(S,T,K) 
&= (T-S)P_t(T)\EE^T\bigl[(F(S,T)-K)^+\big|\cF_t\bigr]	\\
&\leq (T-S)P_t(T)\EE^T\bigl[(R(S,T)-K)^+\big|\cF_t\bigr]
= \PicB_t(S,T,K),
\end{align*}
with the same inequality holding in the case of floorlets.

(2) Always as a consequence of Definition \ref{def:rates}, it holds that
\[
\PicB_t(S,T,K)-\PifB_t(S,T,K)
= \PicF_t(S,T,K)-\PifF_t(S,T,K),
\qquad\text{ for all }t\in[0,S].
\]
This relation implies that, before the beginning of the accrual period, forward-looking and backward-looking caplets/floorlets satisfy the same put-call parity relation.
\end{rem}

\begin{rem}
According to the prevailing ISDA protocol (see \cite{ISDA}), the primary fallback for derivatives based on Libor rates is the compounded overnight risk-free rate with the addition of a {\em credit adjustment spread}. In the case of Libor caplets, this means that a tenor-dependent spread $c(T-S)$ has to be added to the rate in \eqref{eq:general_formulae}. For all combinations of tenors and currencies, the values of the credit adjustment spread have been fixed on 5 March 2021 on the basis of historical data\footnote{See \url{https://www.isda.org/2021/03/05/isda-statement-on-uk-fca-libor-announcement}.}.
Therefore, since the spread $c(T-S)$ is predetermined, the valuation of caplets/floorlets under the ISDA fallback provisions reduces to formulae \eqref{eq:general_formulae}, replacing $K$ with $K-c(T-S)$.

The possibility of a {\em dynamic} credit adjustment spread $(S_t)_{t\geq0}$, such as the across-the-curve credit spread index (AXI) proposed by \cite{BDZ20}, can also be taken into account in our framework. Modeling $(S_t)_{t\geq0}$ as an affine function of $(X_t)_{t\geq0}$, one can derive pricing formulas for options written on credit-adjusted RFRs by a straightforward adaptation of the following theorems.
\end{rem}

Until the end of this section, let us consider fixed but arbitrary $0\leq S<T<+\infty$ and $K>0$. For brevity of notation, we denote
\[
K':=1+(T-S)K.
\]
Moreover, we define  
\[
h^K(w,\lambda) := \frac{1}{2\pi}\frac{(1+(T-S)K)^{w+\im\lambda}}{(w+\im\lambda)(w-1+\im\lambda)},
\qquad\text{ for all $w\in\R\setminus\{0,1\}$ and $\lambda\in\R$}.
\]
Note that the function $\lambda\mapsto h^K(w,\lambda)$ is integrable on $\R$, for every $w\in\R\setminus\{0,1\}$.

The next result presents a general pricing formula for forward-looking caplets/floorlets, based on a Fourier decomposition of the payoff. Since forward-looking caplets/floorlets can be reduced to put/call options on ZCBs (see formula \eqref{eq:fwd_proof} below), the proof follows the well-known structure of \cite[Corollary 10.4]{fil09}, here extended to the case of general affine processes.

\begin{thm}	\label{thm:forward}
There exist constants $w_-<0$ and $w_+>1$ such that the prices of forward-looking caplets and floorlets can be expressed respectively as follows, for all $w\in(w_-,w_+)\setminus\{0,1\}$:
\begin{align}
\PicF_t(S,T,K) &= \begin{cases}
\PiF_t(w),&\text{ for }w\in(w_-,0),	\\
\PiF_t(w)+P_t(S),&\text{ for }w\in(0,1),
\end{cases}
\label{eq:fwd_caplet_Fourier}\\
\PifF_t(S,T,K) &= \begin{cases}
\PiF_t(w),&\text{ for }w\in(1,w_+),	\\
\PiF_t(w)+K'P_t(T),&\text{ for }w\in(0,1),
\end{cases}
\label{eq:fwd_floorlet_Fourier}
\end{align}
for all $t\in[0,S]$, where
\[
\PiF_t(w) := \int_{\R}e^{(w+\im\lambda)A^0(S,T,1)+A^1(t,S,(w+\im\lambda)B^0(T-S,1))+\langle B^1(S-t,(w+\im\lambda)B^0(T-S,1)),X_t\rangle}h^K(w,\lambda)\ud\lambda.
\]
\end{thm}
\begin{proof}
By Assumption \ref{ass:bond}, it holds that $T_+(0,-1)>T$. The semiflow property \eqref{eq:semiflow} implies that $\Psi(S,B^0(T-S,1),-1) = \Psi(T,0,-1)$ and, hence, $(B^0(T-S,1),-1)\in\cY_{S}$. Since  $\cY_{S}$ is open, there exist $w_+>1$ and $w_-<0$ such that $(wB^0(T-S,1),-1)\in\cY_{S}$ for all $w\in(w_-,w_+)$.
In view of Proposition \ref{prop:ext_Fourier}-(i), this implies that
\begin{align*}
&\EE\left[\frac{1}{B_S}\int_{\R}\bigl|P_S(T)^{w+\im\lambda}h^K(w,\lambda)\bigr|\ud\lambda\right]	\\
&\leq \EE\left[\frac{P_S(T)^w}{B_S}\right]\int_{\R}|h^K(w,\lambda)|\ud\lambda	\\
&= \EE\Bigl[e^{-Y_S-L(0,S)+wA^0(S,T,1)+w\langle B^0(T-S,1),X_S\rangle}\Bigr]\int_{\R}|h^K(w,\lambda)|\ud\lambda <+\infty,
\end{align*}
for all $w\in(w_-,w_+)\setminus\{0,1\}$.
We can therefore apply Fubini's theorem and obtain
\begin{align}
&\EE\left[\frac{B_t}{B_S}\int_{\R}P_S(T)^{w+\im\lambda}h^K(w,\lambda)\ud\lambda\bigg|\cF_t\right]	\notag\\
&= \int_{\R}\EE\left[\frac{B_t}{B_S}P_S(T)^{w+\im\lambda}\bigg|\cF_t\right]h^K(w,\lambda)\ud\lambda	
\label{eq:fwd_Fubini}\\
&= \int_{\R}\EE\Bigl[e^{-(Y_S-Y_t)-L(t,S)+(w+\im\lambda)A^0(S,T,1)+(w+\im\lambda)\langle B^0(T-S,1),X_S\rangle}\Big|\cF_t\Bigr]h^K(w,\lambda)\ud\lambda
\notag
=\PiF_t(w),
\end{align}
where the last equality follows from an application of formula \eqref{eq:ext_Fourier}, which is justified due to the fact that $((w+\im\lambda)B^0(T-S,1),-1)\in\cD_S\subseteq\cD_{S-t}$.
The result then follows by noting that
\be	\label{eq:fwd_proof}
\PicF_t(S,T,K)
= \EE\biggl[\frac{B_t}{B_T}\left(\frac{1}{P_S(T)}-K'\right)^+\bigg|\cF_t\biggr]
= \EE\left[\frac{B_t}{B_S}\left(1-K'P_S(T)\right)^+\bigg|\cF_t\right]
\ee
and applying Lemma \ref{lem:Fourier_payoffs} to the last term on the right-hand side of \eqref{eq:fwd_proof}, making use of \eqref{eq:fwd_Fubini}. 
In the case of a forward-looking floorlet, formula \eqref{eq:fwd_floorlet_Fourier} can be proved in an analogous way.
\end{proof}

The valuation of backward-looking caplets/floorlets requires a more delicate analysis, mainly as a consequence of the different measurability properties of their payoffs. 
In particular, we shall make use of the following integrability property of a functional of the affine process.

\begin{lem}	\label{lem:Holder}
For every $0\leq S\leq T<+\infty$, there exist constants $\eta_1,\eta_2\in(0,1)$ such that
\be	\label{eq:Holder}
\EE\bigl[e^{-(1+\eta_1)Y_S+\eta_1(Y_T-Y_S)}\bigr] < +\infty
\qquad\text{and}\qquad
\EE\bigl[e^{-(1-\eta_2)Y_S-(1+\eta_2)(Y_T-Y_S)}\bigr] < +\infty.
\ee
Moreover, the constants $\eta_1$ and $\eta_2$ do not depend on $S$.
\end{lem}
\begin{proof}
Since $\cY_T$ is open and $(0,-1)\in\cY_T$, due to Assumption \ref{ass:bond}, there exists $\varepsilon>0$ such that $(0,-1-\varepsilon)\in\cY_T$. Moreover, since $(0,0)\in\cY_T$ always holds, we can assume that $(0,\varepsilon)\in\cY_T$, for $\varepsilon$ small enough. Taking $\eta_1=\varepsilon^2/(1+3\varepsilon)$, H\"older's inequality with $p=(1+3\varepsilon)/(1+2\varepsilon)$ and $q=p/(p-1)$ implies that
\begin{align*}
\EE\bigl[e^{-(1+\eta_1)Y_S+\eta_1(Y_T-Y_S)}\bigr] 
&= \EE\bigl[e^{-(1+2\eta_1)Y_S+\eta_1 Y_T}\bigr] \\
&\leq  \EE\bigl[e^{-p(1+2\eta_1)Y_S}\bigr]^{\frac{1}{p}}
\EE\bigl[e^{q\eta_1Y_T}\bigr]^{\frac{1}{q}}
= \EE\bigl[e^{-(1+\varepsilon)Y_S}\bigr]^{\frac{1}{p}}
\EE\bigl[e^{\varepsilon Y_T}\bigr]^{\frac{1}{p}} < +\infty,
\end{align*}
as a consequence of part (i) of Proposition \ref{prop:ext_Fourier}, together with the fact that $\cY_S\subseteq\cY_T$, for every $S\in[0,T]$.
Similarly, taking $\eta_2=\varepsilon^2/(2+3\varepsilon)$, $p=(2+3\varepsilon)/(2\varepsilon)$ and $q=p/(p-1)$, we have that
\begin{align*}
\EE\bigl[e^{-(1-\eta_2)Y_S-(1+\eta_2)(Y_T-Y_S)}\bigr] 
&= \EE\bigl[e^{2\eta_2Y_S-(1+\eta_2)Y_T}\bigr] \\
&\leq  \EE\bigl[e^{2p\eta_2Y_S}\bigr]^{\frac{1}{p}}
\EE\bigl[e^{-q(1+\eta_2)Y_T}\bigr]^{\frac{1}{q}}
= \EE\bigl[e^{\varepsilon Y_S}\bigr]^{\frac{1}{p}}
\EE\bigl[e^{-(1+\varepsilon)Y_T}\bigr]^{\frac{1}{p}} < +\infty,
\end{align*}
again by part (i) of Proposition \ref{prop:ext_Fourier}. 
\end{proof}

\begin{rem}
As a consequence of Lemma \ref{lem:Holder}, the moment generating function $\EE[\exp(u Y_T)]$ is well-defined for all $u\in[-(1+\eta_2),\eta_1]$, for suitable constants $\eta_1$ and $\eta_2$ (depending on $T$). 
This avoids the potential explosion of the moment generating function, which is a well-known phenomenon for some non-Gaussian affine processes (see, e.g.,  \cite[Corollary 3.3]{AP07} when $X$ is a square-root diffusion). 
The constants $\eta_1$ and $\eta_2$ can be explicitly determined by relying on \cite[Theorem 3.4]{FGSz22} when $X$ is an affine process taking values in $\R_+$ (CBI process).
\end{rem}

By relying on Lemma \ref{lem:Holder}, we are now in a position to state the following theorem, which provides a general Fourier-based pricing formula for backward-looking caplets/floorlets.
In the next theorem, we consider the case of valuation before the accrual period (i.e., for $t\leq S$), referring to Remark \ref{rem:inside_accrual} for the valuation inside the accrual period (i.e., for $t\in[S,T]$). 

\begin{thm}	\label{thm:backward}
There exist constants $w_-<0$ and $w_+>1$ such that the prices of backward-looking caplets and floorlets can be expressed respectively as follows, for all $w\in(w_-,w_+)\setminus\{0,1\}$:
\begin{align}
\PicB_t(S,T,K) &= \begin{cases}
\PiB_t(w),&\text{ for }w\in(w_-,0),	\\
\PiB_t(w)+P_t(S),&\text{ for }w\in(0,1),
\end{cases}
\label{eq:bck_caplet_Fourier}\\
\PifB_t(S,T,K) &= \begin{cases}
\PiB_t(w),&\text{ for }w\in(1,w_+),	\\
\PiB_t(w)+K'P_t(T),&\text{ for }w\in(0,1),
\end{cases}
\label{eq:bck_floorlet_Fourier}
\end{align}
for all $t\in[0,S]$, where
\[
\PiB_t(w) := \int_{\R}e^{A^0(S,T,w+\im\lambda)+A^1(t,S,B^0(T-S,w+\im\lambda))+\langle B^1(S-t,B^0(T-S,w+\im\lambda)),X_t\rangle}h^K(w,\lambda)\ud\lambda.
\]
\end{thm}
\begin{proof}
We first show that there exist constants $w_-<0$ and $w_+>1$ such that
\be	\label{eq:w_backward}
(B^0(T-S,w_-),-1)\in\cY_S
\qquad\text{ and }\qquad
(B^0(T-S,w_+),-1)\in\cY_S.
\ee
By \eqref{eq:Holder} together with \cite[Theorem 2.14-(a)]{krm12}, the ODEs \eqref{eq:Riccati1}-\eqref{eq:Riccati2} admit a solution up to time $S$ starting from $(u,v)=(\Psi(T-S,0,\eta_1),-(1+\eta_1))$ (however, note that it may happen that $(\Psi(T-S,0,\eta_1),-(1+\eta_1))\in\partial\cY_S$). In addition, since $(0,\eta_1)\in\cY_T$ by Lemma \ref{lem:Holder}, the semiflow property \eqref{eq:semiflow} implies that $(\Psi(T-S,0,\eta_1),\eta_1)\in\cY_S$. Let $\gamma:=\eta_1/(1+2\eta_1)\in(0,1)$. Noting that $\gamma\eta_1-(1-\gamma)(1+\eta_1)=-1$, the convexity of $\cY_S$ implies that
\[
(\Psi(T-S,0,\eta_1),-1) = \gamma(\Psi(T-S,0,\eta_1),\eta_1)+(1-\gamma)(\Psi(T-S,0,\eta_1),-(1+\eta_1)) \in\cY_S.
\]
Letting $w_-:=-\eta_1$ proves the first property in  \eqref{eq:w_backward}. Arguing in a similar way, \eqref{eq:Holder} implies the existence of a solution up to time $S$ of the ODEs \eqref{eq:Riccati1}-\eqref{eq:Riccati2} starting from $(u,v)=(\Psi(T-S,0,-(1+\eta_2)),-(1-\eta_2))$. Moreover, it holds that $(\Psi(T-S,0,-(1+\eta_2)),-(1+\eta_2))\in\cY_S$ and convexity of $\cY_S$ then implies that $(\Psi(T-S,0,-(1+\eta_2)),-1)\in\cY_S$, thus proving the second property in \eqref{eq:w_backward} with $w_+:=1+\eta_2$.
Moreover, Lemma \ref{lem:Y_t}  and Lemma \ref{lem:convex} together imply that $(B^0(T-S,w),-1)\in\cY_S$ for all $w\in(w_-,w_+)$.
By Proposition \ref{prop:ext_Fourier}, it follows that
\be	\label{eq:bck_Tonelli}	\begin{aligned}
&\EE\left[\frac{1}{B_S}\int_{\R}\biggl|\left(\frac{B_S}{B_T}\right)^{w+\im\lambda}h^K(w,\lambda)\biggr|\ud\lambda\right]	\\
&\leq \EE\left[\frac{1}{B_S}\left(\frac{B_S}{B_T}\right)^w\right]\int_{\R}|h^K(w,\lambda)|\ud\lambda	\\
&= \EE\left[e^{-Y_S-L(0,S)+A^0(S,T,w)+\langle B^0(T-S,w),X_S\rangle}\right]\int_{\R}|h^K(w,\lambda)|\ud\lambda <+\infty,
\end{aligned}	\ee
for all $w\in(w_-,w_+)\setminus\{0,1\}$.
We can therefore apply Fubini's theorem and obtain
\begin{align}
& \EE\left[\frac{B_t}{B_S}\int_{\R}\left(\frac{B_S}{B_T}\right)^{w+\im\lambda}h^K(w,\lambda)\ud\lambda\bigg|\cF_t\right]	\notag\\
&= \int_{\R}\EE\left[\frac{B_t}{B_S}\left(\frac{B_S}{B_T}\right)^{w+\im\lambda}\bigg|\cF_t\right]h^K(w,\lambda)\ud\lambda	
\label{eq:bck_Fubini}\\
&= \int_{\R}\EE\left[e^{-(Y_S-Y_t)-L(t,S)+A^0(S,T,w+\im\lambda)+\langle B^0(T-S,w+\im\lambda),X_S\rangle}\Big|\cF_t\right]h^K(w,\lambda)\ud\lambda
=\PiB_t(w),
\notag
\end{align}
where the second equality follows from formula \eqref{eq:ext_Fourier} using that $(0,-(w+\im\lambda))\in\cD_{T-S}\subseteq\cD_T$, for every $w\in(w_-,w_+)$, which is in turn a consequence of $(0,-w_-)\in\cY_T$ and $(0,-w_+)\in\cY_T$ together with the convexity of $\cY_T$. To justify the last equality in \eqref{eq:bck_Fubini}, note that, in view of Remark \ref{rem:structure_psi} and arguing as in the proof of \cite[Proposition 5.1]{krm12}, it holds that
\[
\real(B^0_I(T-S,w+\im\lambda))\preceq B^0_I(T-S,w)
\qquad\text{and}\qquad
\real(B^0_J(T-S,w+\im\lambda))=B^0_J(T-S,w).
\]
Recalling that $(B^0(T-S,w),-1)\in\cY_S$ for all $w\in(w_-,w_+)$, Lemma \ref{lem:Y_t} therefore implies that $(\real(B^0(T-S,w+\im\lambda)),-1)\in\cY_S$. The last equality in \eqref{eq:bck_Fubini} then follows by applying formula \eqref{eq:ext_Fourier} with $(u,v)=(B^0(T-S,w+\im\lambda),-1)$.
Formula \eqref{eq:bck_caplet_Fourier} follows by noting that
\be	\label{eq:bck_proof}
\PicB_t(S,T,K)
= \EE\biggl[\frac{B_t}{B_T}\left(\frac{B_T}{B_S}-K'\right)^+\bigg|\cF_t\biggr]
= \EE\left[\frac{B_t}{B_S}\left(1-K'\frac{B_S}{B_T}\right)^+\bigg|\cF_t\right]
\ee
and applying Lemma \ref{lem:Fourier_payoffs} to the last term on the right-hand side of \eqref{eq:bck_proof}, making use of \eqref{eq:bck_Fubini}. 
In the case of a backward-looking floorlet, formula \eqref{eq:bck_floorlet_Fourier} can be proved in an analogous way.
\end{proof}

\begin{rem}	\label{rem:inside_accrual}
Since the payoff of backward-looking caplets/floorlets is $\cF_T$-measurable and not $\cF_S$-measurable, one may also consider the valuation of such products inside the accrual period $[S,T]$. In this case,  the price of a backward-looking caplet  for $t\in[S,T]$ can be expressed as 
\be	\label{eq:inside_period}	\begin{aligned}
\PicB_t(S,T,K)
&= \EE\biggl[\left(\frac{B_t}{B_S}-K'\frac{B_t}{B_T}\right)^+\bigg|\cF_t\biggr]	\\
&= \int_{\R}\EE\biggl[\left(\frac{B_t}{B_T}\right)^{w+\im\lambda}\bigg|\cF_t\biggr]\left(\frac{B_t}{B_S}\right)^{-(w-1+\im\lambda)}h^K(w,\lambda)\ud\lambda \\
&= \int_{\R}e^{A^0(t,T,w+\im\lambda)+\langle B^0(T-t,w+\im\lambda),X_t\rangle}
\left(\frac{B_t}{B_S}\right)^{-(w-1+\im\lambda)}h^K(w,\lambda)\ud\lambda,
\end{aligned}	\ee
for any $w\in(w_-,0)$.
The application of Fubini's theorem in \eqref{eq:inside_period} is justified by \eqref{eq:bck_Tonelli}, while the last equality follows from formula \eqref{eq:ext_Fourier} using the fact that $(0,-w-\im\lambda)\in\cD_T$ for all $w\in(w_-,w_+)$, as shown in the proof of Theorem \ref{thm:backward}.
\end{rem}

As pointed out in \cite{LM19}, in a market where forward-looking and backward-looking rates are present, one may also consider a {\em term-basis caplet}, corresponding to an exchange option between forward-looking and backward looking rates.\footnote{A term-basis caplet can also be viewed as a forward-start option (see \cite[Section 6.2]{MR05}) on the backward-looking rate, with strike set equal to the forward-looking rate at the beginning of the accrual period.} 
The payoff of a term-basis caplet is given by $(T-S)(R(S,T)-F(S,T))^+$ and the corresponding arbitrage-free price is
\[
\pi^{c,{\rm TB}}_t(S,T) 
= (T-S)\EE\left[\frac{B_t}{B_T}\bigl(R(S,T)-F(S,T)\bigr)^+\bigg|\cF_t\right],
\qquad\text{ for }t\in[0,T].
\]
Observe that the price of a term-basis caplet equals the price of the corresponding term-basis floorlet for all $t\in[0,S]$. This fact is a direct consequence of the definition of forward-looking rate. On the other hand, for $t\in[S,T]$, term-basis caplets/floorlets reduce to backward-looking caplets/floorlets that can be priced as in Remark \ref{rem:inside_accrual}. In the following, we therefore restrict our attention to the valuation of term-basis caplets on $[0,S]$.
As a preliminary, let us define
\[
k(w,\lambda) := \frac{1}{2\pi}\frac{1}{(w+\im\lambda)(w-1+\im\lambda)},
\qquad\text{ for all $w\in\R\setminus\{0,1\}$ and $\lambda\in\R$}.
\]
Similarly as above, the next theorem relies on a Fourier representation of the payoff of a term-basis caplet. The proof requires a specific analysis of the integrability properties of the terms appearing in the representation and of the applicability of the affine transform formula \eqref{eq:ext_Fourier}.

\begin{thm}	\label{thm:term_basis}
There exists a constant $w_-<0$ such that the price of a term-basis caplet can be expressed as follows, for all $w\in(w_-,1)\setminus\{0\}$:
\be	\label{eq:term_basis}
\pi^{c,{\rm TB}}_t(S,T)  = \begin{cases}
\Pi_t^{{\rm TB}}(w),&\text{ for }w\in(w_-,0),	\\
\Pi_t^{{\rm TB}}(w)+P_t(S),&\text{ for }w\in(0,1),
\end{cases}
\ee
for all $t\in[0,S]$, where
\begin{align*}
\Pi_t^{{\rm TB}}(w) &:= \int_{\R}e^{A^1(t,S,B^0(T-S,w+\im\lambda)-(w+\im\lambda)B^0(T-S,1))+A^0(S,T,w+\im\lambda)-(w+\im\lambda)A^0(S,T,1)}\\
&\qquad\quad\times e^{\langle B^1(S-t,B^0(T-S,w+\im\lambda)-(w+\im\lambda)B^0(T-S,1)),X_t\rangle}k(w,\lambda)\ud\lambda.
\end{align*}
\end{thm}
\begin{proof}
Observe that the price $p^{c,{\rm TB}}_t(S,T)$ of a term-basis caplet can be expressed as follows:
\be	\label{eq:TB_pricing}
\pi^{c,{\rm TB}}_t(S,T) 
= \EE\biggl[\frac{B_t}{B_T}\left(\frac{B_T}{B_S}-\frac{1}{P_S(T)}\right)^+\bigg|\cF_t\biggr]	
= \EE\biggl[\frac{B_t}{B_S}\left(1-\frac{B_S}{B_TP_S(T)}\right)^+\bigg|\cF_t\biggr].
\ee
In order to apply Lemma \ref{lem:Fourier_payoffs}, we first note that
\be	\label{eq:TB_Tonelli}
\EE\biggl[\frac{1}{B_S}\int_{\R}\biggl|\left(\frac{B_S}{B_TP_S(T)}\right)^{w+\im\lambda}k(w,\lambda)\biggr|\ud\lambda\biggr]
\leq \EE\left[\frac{1}{B_S}\left(\frac{B_S}{B_TP_S(T)}\right)^w\right]\int_{\R}|k(w,\lambda)|\ud\lambda.
\ee
The expected value on the right-hand side of the above inequality is obviously finite for $w=0$ and $w=1$ and, therefore, for all $w\in(0,1)$ by convexity. 
Moreover, letting $\eta_1\in(0,1)$ be the constant appearing in Lemma \ref{lem:Holder},  an application of H\"older's inequality yields
\be	\label{eq:TB_Holder}
\EE\left[\frac{1}{B_S}\left(\frac{B_S}{B_TP_S(T)}\right)^{-\frac{\eta_1}{2}}\right]
\leq \EE\left[\frac{1}{B_S}\left(\frac{B_S}{B_T}\right)^{-\eta_1}\right]^{1/2}
\EE\left[\frac{1}{B_S}P_S(T)^{\eta_1}\right]^{1/2}.
\ee
Recall from the proof of Theorem \ref{thm:backward} that $(\Psi(T-S,0,\eta_1),-1)\in\cY_S$. Therefore, the first expectation on the right-hand side of \eqref{eq:TB_Holder} is finite-valued, while for the second we have
\[
\EE\left[\frac{1}{B_S}P_S(T)^{\eta_1}\right]
\leq \EE\left[\frac{1}{B_S}\bigl(\eta_1P_S(T)+1-\eta_1\bigr)\right]
= \eta_1P_0(T)+(1-\eta_1)P_0(S)<+\infty,
\]
where we have used the convexity of the function $\alpha\mapsto P_S(T)^\alpha$ and Assumption \ref{ass:bond}.
Therefore, taking $w_-:=-\eta_1/2$, we have that \eqref{eq:TB_Tonelli} is finite-valued for all $w\in(-w_-,1)\setminus\{0\}$.
We can therefore apply Fubini's theorem and obtain
\be	\label{eq:TB_Fubini}\begin{aligned}
&\EE\biggl[\frac{B_t}{B_S}\int_{\R}\left(\frac{B_S}{B_TP_S(T)}\right)^{w+\im\lambda}k(w,\lambda)\ud\lambda\bigg|\cF_t\biggr]	\\
& = \int_{\R}\EE\biggl[\frac{B_t}{B_S}\left(\frac{B_S}{B_TP_S(T)}\right)^{w+\im\lambda}\bigg|\cF_t\biggr]k(w,\lambda)\ud\lambda	\\
& = \int_{\R}\EE\left[\frac{B_t}{B_S}e^{A^0(S,T,w+\im\lambda)+\langle B^0(T-S,w+\im\lambda),X_S\rangle}P_S(T)^{-(w+\im\lambda)}\bigg|\cF_t\right]k(w,\lambda)\ud\lambda	\\
& = \int_{\R}e^{-L(t,S)+A^0(S,T,w+\im\lambda)-(w+\im\lambda)A^0(S,T,1)}\\
&\qquad\times\EE\left[e^{-(Y_S-Y_t)+\langle B^0(T-S,w+\im\lambda)-(w+\im\lambda)B^0(T-S,1),X_S\rangle}\Big|\cF_t\right]k(w,\lambda)\ud\lambda
= \Pi_t^{{\rm TB}}(w).
\end{aligned}\ee
In \eqref{eq:TB_Fubini}, the second equality follows from an application of formula \eqref{eq:ext_Fourier}, using the fact that $(0,-(w+\im\lambda))\in\cD_{T-S}\subseteq\cD_T$ for all $w\in(w_-,1)\setminus\{0\}$ (compare with the proof of Theorem \ref{thm:backward}). 
To justify the last equality in \eqref{eq:TB_Fubini}, note first that $B^0_J(T-S,w)-wB^0_J(T-S,1)=0$ (see Remark \ref{rem:structure_psi}). 
For all $w\in(0,1)$, Lemma \ref{lem:convex} implies that $B^0_I(T-S,w)-wB^0_I(T-S,1)\preceq0$. Lemma \ref{lem:Y_t} and Assumption \ref{ass:bond} then imply $(B^0(T-S,w)-wB^0(T-S,1),-1)\in\cY_S$, for all $w\in(0,1)$. Similarly as in the proof of \cite[Proposition 5.1]{krm12}, we have that $\real(B^0_I(T-S,w+\im\lambda))\preceq B^0_I(T-S,w)$ and a further application of Lemma \ref{lem:Y_t} yields $(B^0(T-S,w+\im\lambda)-(w+\im\lambda)B^0(T-S,1),-1)\in\cD_S$. This justifies the last equality in \eqref{eq:TB_Fubini} for all $w\in(0,1)$.
Considering now the case $w\in(w_-,0)$, note that, since \eqref{eq:TB_Holder} is finite, there exists a solution up to time $S$ to \eqref{eq:Riccati1}-\eqref{eq:Riccati2} starting from $(B^0(T-S,w_-)-w_-B^0(T-S,1),-1)$, but this point may lie on the boundary $\partial\cY_S$.
However, setting $\gamma:=w/w_-\in(0,1)$, again by Lemma \ref{lem:convex} we have that
\begin{align*}
B^0_I(T-S,w)-wB^0_I(T-S,1)
&= B^0_I(T-S,\gamma w_-)-\gamma w_-B^0_I(T-S,1)	\\
&\preceq \gamma\left(B^0_I(T-S,w_-)-w_-B^0_I(T-S,1)\right).
\end{align*}
In particular, since $(0,-1)\in\cY_S$ by Assumption \ref{ass:bond} and the set $\cY_S$ is convex, this implies that $(B^0(T-S,w)-wB^0(T-S,1),-1)\in\cY_S$, for all $w\in(w_-,0)$. Similarly as above, this yields $(B^0(T-S,w+\im\lambda)-(w+\im\lambda)B^0(T-S,1),-1)\in\cD_S$, thus completing the proof of \eqref{eq:TB_Fubini}.
Formula \eqref{eq:term_basis} then follows by applying Lemma \ref{lem:Fourier_payoffs} to equation \eqref{eq:TB_pricing}, making use of \eqref{eq:TB_Fubini}.
\end{proof}

\begin{rem}
It is interesting to remark that the pricing function $\Pi^{\rm TB}_t(w)$ in Theorem \ref{thm:term_basis} does only depend on the first $m$ components of the affine process $X$. This follows from the fact that $B^0_J(T-S,w+\im\lambda)-(w+\im\lambda)B^0_J(T-S,1)=0$ (see Remark \ref{rem:structure_psi} and the proof of Theorem \ref{thm:term_basis}).
\end{rem}

\begin{rem}	\label{rem:comments}
In this remark, we discuss briefly some general aspects that are relevant for the numerical implementation of the pricing formulas derived above:
\begin{enumerate}
\item When computing the quantities $\PiF_t(w)$ and $\PiB_t(w)$ that appear in Theorems \ref{thm:forward} and \ref{thm:backward}, respectively, the parts that depend on the characteristic function have to be evaluated only once for all different strikes. This fact is important for model calibration, especially in models for which the Riccati ODEs \eqref{eq:Riccati1}-\eqref{eq:Riccati2} have to be solved numerically. 
\item The computation of the quantities $\PiF_t(w)$, $\PiB_t(w)$ and $\Pi^{{\rm TB}}_t(w)$ is subject to a truncation error (due to the truncation of the integral at some suitably chosen upper/lower bounds) and a discretization error (due to approximating the integral with a finite sum). Explicit error bounds that are applicable to our setting have been derived in \cite{lee2004}.
\item The choice of the parameter $w$ depends on the specific properties of the affine process under consideration and has to be suitably chosen in order to exclude large oscillations of the integrand. The choice of $w$ can also depend on the maturity and the strike of the products to be priced (see, e.g., \cite{lee2004}). As can be seen from the proofs above, the range of possible values of $w$ is crucially related to the integrability properties of the model.
Precise recommendations on the choice of $w$, as well as of the other numerical parameters, are given in \cite{leven2016}, in an affine setup that also covers our setting. 
\end{enumerate}
\end{rem}

The pricing formulae stated in Theorems \ref{thm:forward}, \ref{thm:backward} and \ref{thm:term_basis} are based on Fourier decomposition of the payoffs. Alternative pricing formulae can be obtained by passing to the $S$-forward and $T$-forward measures, exploiting the fact that the characteristic function of the affine process $X$ under any forward measure can be explicitly determined (compare with \cite[Corollary 10.2]{fil09}). 
We illustrate this methodology in the case of forward-looking and backward-looking caplets.

\begin{prop}	\label{prop:fwd_measures}
It holds that
\begin{align}
\PicF_t(S,T,K) &= P_t(S)p^S_t(\cI^{\rm F})-K'P_t(T)p^T_t(\cI^{\rm F}),
\label{eq:fwd_caplet_proba}\\
\PicB_t(S,T,K) &= P_t(S)q^S_t(\cI^{\rm B})-K'P_t(T)q^T_t(\cI^{\rm B}),
\label{eq:bck_caplet_proba}
\end{align}
where  $\cI^{\rm F}:=(-\infty,-A^0(S,T,1)-\log(K'))$, $\cI^{\rm B}:=(\log(K')-L(S,T),+\infty)$ and
\begin{itemize}
\item $p_t^S(\ud y)$ and $p_t^T(\ud y)$ denote respectively the $\cF_t$-conditional distributions of the random variable $\langle B^0(T-S,1),X_S\rangle$ under the $S$-forward and $T$-forward measures;
\item $q_t^S(\ud y)$ and $q_t^T(\ud y)$ denote respectively the $\cF_t$-conditional distributions of the random variable $Y_T-Y_S$ under the $S$-forward and $T$-forward measures.
\end{itemize}  
\end{prop}
\begin{proof}
Formula \eqref{eq:fwd_caplet_proba} is a direct consequence of \eqref{eq:fwd_proof} and \eqref{eq:ZCB_affine}, while formula \eqref{eq:bck_caplet_proba} follows from \eqref{eq:bck_proof} together with the fact that $B_T/B_S=\exp(L(S,T)+(Y_T-Y_S))$.
\end{proof}

The distributions appearing in Proposition \ref{prop:fwd_measures} can be recovered from the $\cF_t$-conditional characteristic function of $(\langle B^0(T-S,1),X_S\rangle,Y_T-Y_S)$ under the $S$-forward and $T$-forward measures. In view of Proposition \ref{prop:ext_Fourier}, the latter can be computed as follows, for all $(\zeta_1,\zeta_2)\in\R^2$:
\begin{align*}
& \EE^S\bigl[e^{\im\zeta_1\langle B^0(T-S,1),X_S\rangle+\im\zeta_2(Y_T-Y_S)}\big|\cF_t\bigr]	\\
&= \frac{1}{P_t(S)}\EE\bigl[e^{-L(t,S)-(Y_S-Y_t)+\im\zeta_1\langle B^0(T-S,1),X_S\rangle+\im\zeta_2(Y_T-Y_S)}\big|\cF_t\bigr]	\\
&= \frac{1}{P_t(S)}\EE\bigl[e^{-L(t,S)-\im\zeta_2L(S,T)+A^0(S,T,-\im\zeta_2)-(Y_S-Y_t)+\langle \im\zeta_1B^0(T-S,1)+B^0(T-S,-\im\zeta_2),X_S\rangle}\big|\cF_t\bigr]	\\
&= \frac{1}{P_t(S)}e^{-\im\zeta_2L(S,T)+A^0(S,T,-\im\zeta_2)+A^1(t,S,\im\zeta_1B^0(T-S,1)+B^0(T-S,-\im\zeta_2))+\langle B^1(S-t,\im\zeta_1B^0(T-S,1)+B^0(T-S,-\im\zeta_2)),X_t\rangle}
\end{align*}
and
\begin{align*}
& \EE^T\bigl[e^{\im\zeta_1\langle B^0(T-S,1),X_S\rangle+\im\zeta_2(Y_T-Y_S)}\big|\cF_t\bigr]	\\
&= \frac{1}{P_t(T)}\EE\bigl[e^{-L(t,T)-(Y_S-Y_t)+\im\zeta_1\langle B^0(T-S,1),X_S\rangle+(\im\zeta_2-1)(Y_T-Y_S)}\big|\cF_t\bigr]	\\
&= \frac{1}{P_t(T)}\EE\bigl[e^{-L(t,S)-\im\zeta_2L(S,T)+A^0(S,T,1-\im\zeta_2)-(Y_S-Y_t)+\langle \im\zeta_1B^0(T-S,1)+B^0(T-S,1-\im\zeta_2),X_S\rangle}\big|\cF_t\bigr]	\\
&= \frac{1}{P_t(T)}e^{-\im\zeta_2L(S,T)+A^0(S,T,1-\im\zeta_2)+A^1(t,S,\im\zeta_1B^0(T-S,1)+B^0(T-S,1-\im\zeta_2))+\langle B^1(S-t,\im\zeta_1B^0(T-S,1)+B^0(T-S,1-\im\zeta_2)),X_t\rangle}.
\end{align*}

\begin{rem}	\label{rem:fwd_measures}
For some specific models, the conditional distributions appearing in Proposition \ref{prop:fwd_measures} can be explicitly computed. In particular, this is the case for (multi-factor) Hull-White models, which have the property of preserving the Gaussian distribution of $X_S$ and $Y_T-Y_S$ under any forward measure. In this setting, one can deduce from Proposition \ref{prop:fwd_measures} the caplet pricing formulas recently derived in \cite{Hasegawa21,Hof20,RB21,Turf21,Xu22}.
\end{rem}

\section{An example: CIR++ model for an RFR process}
\label{sec:example}

In this section, we illustrate the applicability of Proposition \ref{prop:fwd_measures} in the context of the CIR++ model introduced in \cite{BM01}. 
We assume that the RFR process is given by $r:=\ell(\cdot)+X$, where $X=(X_t)_{t\geq0}$ is a square-root process:
\be	\label{eq:CIR}
\ud X_t = (b - \beta X_t)\ud t + \sigma\sqrt{X_t}\ud W_t,
\qquad X_0>0,
\ee
where $b,\beta,\sigma>0$. The explicit expression of the unique function $\ell$ that fits the term structure at $t=0$ is given in \cite[Section 6]{BM01}.
For the CIR++ model, the explicit form of the functions $A^0(t,T,v)$ and $B^0(T-t,v)$ can be deduced from \cite[Corollary 6.3.4.2]{JYC}:
\begin{align*}
A^0(t,T,v) &= \frac{2b}{\sigma^2}\log\left(\frac{2\theta_v e^{\frac{(\theta_v+\beta)(T-t)}{2}}}{2\theta_v+(\beta+\theta_v)(e^{\theta_v(T-t)}-1)}\right) - v\int_t^T\ell(u)\ud u,\\
B^0(T-t,v) &= \frac{-2v}{\beta+\theta_v\coth(\frac{\theta_v(T-t)}{2})},
\end{align*}
with $\theta_v:=\sqrt{\beta^2+2v\sigma^2}$. Setting $v=1$ enables us to explicitly compute ZCB prices by \eqref{eq:ZCB_affine}.
\cite[Proposition 6.3.4.1]{JYC} implies that the functions $A^1(t,T,u)$ and $B^1(T-t,u)$ are given by
\begin{align*}
A^1(t,T,u) &= \frac{2b}{\sigma^2}\log\left(\frac{2\theta e^{\frac{(\theta+\beta)(T-t)}{2}}}{\theta(e^{\theta(T-t)}+1)+\beta(e^{\theta(T-t)}-1)-u\sigma^2(e^{\theta(T-t)}-1)}\right) 
-\int_t^T\ell(u)\ud u,\\
B^1(T-t,u) &= \frac{u(\theta+\beta+e^{\theta(T-t)}(\theta-\beta))-2(e^{\theta(T-t)}-1)}{\theta(e^{\theta(T-t)}+1)+\beta(e^{\theta(T-t)}-1)-u\sigma^2(e^{\theta(T-t)}-1)},
\end{align*}
where we denote $\theta:=\theta_1$ for brevity of notation.

In the CIR++ model, the price of a forward-looking caplet can be computed in closed form, since the $\cF_t$-conditional distributions $p^S_t(\ud y)$ and $p^T_t(\ud y)$ appearing in formula \eqref{eq:fwd_caplet_proba} can be explicitly determined. Indeed, for all $0\leq t\leq S\leq T<+\infty$, the $\cF_t$-conditional density $p^T_{(t,S)}(x)$ of $X_S$ under the $T$-forward measure is given by (see, e.g., \cite[Section 6]{BM01})
\[
p^T_{(t,S)}(x) = f_{\chi^2(\nu,\delta(t,S)X_t)}(x),
\]
where $f_{\chi^2(\nu,\delta(t,S)X_t)}$ denotes the density function of a non-central $\chi^2$ distribution with $\nu$ degrees of freedom and non-centrality parameter $\delta(t,S)X_t$, with $\nu:=4b/\sigma^2$ and
\[
\delta(t,S) := \frac{4\rho^2(S-t)e^{\theta(S-t)}}{4(\rho(S-t)+(\beta+\theta)/\sigma^2-B^0(T-S,1))^2},
\quad\text{ where }\quad
\rho(S-t) := \frac{2\theta}{\sigma^2(e^{\theta(S-t)}-1)}.
\]

The price of a backward-looking caplet can be computed by relying on formula \eqref{eq:bck_caplet_proba}. For a square-root process $X$, the $\cF_t$-conditional distribution of $\int_S^TX_u\ud u$ under a forward measure is not known. However, it can be retrieved by the Gil-Pelaez  inversion formula (see \cite{GP51}), making use of the explicit knowledge of the $\cF_t$-conditional characteristic function of $\int_S^TX_u\ud u$ under any forward measure. This leads to the following semi-closed pricing formula:
\[
\PicB_t(S,T,K)
= \frac{P_t(S)-K'P_t(T)}{2}
+\frac{1}{\pi}\int_0^{+\infty}\frac{\imag\left((K')^{-\im x}(e^{g_1(x)}-K'e^{g_2(x)})\right)}{x}\ud x,
\]
where
\begin{align*}
g_1(x) &:= A^0(S,T,-\im x)+A^1(t,S,B^0(T-S,-\im x))+\langle B^1(S-t,B^0(T-S,-\im x)),X_t\rangle,\\
g_2(x) &:= A^0(S,T,1-\im x)+A^1(t,S,B^0(T-S,1-\im x))+\langle B^1(S-t,B^0(T-S,1-\im x)),X_t\rangle.
\end{align*}

\begin{rem}	\label{rem:Wishart}
The CIR++ model of this section can be generalized to a Wishart driving process, as considered in \cite{Gno12}. By \cite[Lemma 5.3]{CFG19}, a Wishart process has a non-central Wishart distribution with known parameters under any forward measure. Similarly as above, this allows for the explicit computation of the conditional probabilities appearing in formula \eqref{eq:fwd_caplet_proba}, while the conditional probabilities appearing in formula \eqref{eq:bck_caplet_proba} can be recovered by Fourier inversion.
\end{rem}
\vspace{-0.3cm}

\section{Pricing of futures contracts}	\label{sec:futures}

In this section, we show that 1-month (1M) and 3-month (3M) futures contracts can be efficiently priced in the context of the affine RFR framework introduced in Section \ref{sec:model}. This is especially important in view of calibration of a model given that futures contracts are currently the most liquid RFR-based products. 
As explained in \cite{Merc18,LM19}, 1M and 3M futures contracts are characterized by different settlement specifications. 
A 3M futures contract settles at $T$ at the backward-looking rate $R(S,T)$ (representing the geometric average of overnight rates over the period $[S,T]$, with $T-S$ being equal to three months). In this case, denoting by $f^{{\rm 3M}}(t,S,T)$ the 3M futures rates at time $t$, it holds that
\be	\label{eq:future_3M}
f^{{\rm 3M}}(t,S,T) 
= \EE[R(S,T)|\cF_t]
= \frac{1}{T-S}\left(e^{L(S,T)}\EE[e^{Y_T-Y_S}|\cF_t]-1\right),
\ee
for $0\leq t\leq S<T<+\infty$. As long as $(0,1)\in\cY_T$, the  conditional expectation in the right-hand side of \eqref{eq:future_3M} can be explicitly computed by a direct application of Proposition \ref{prop:ext_Fourier}.

\begin{rem}	\label{rem:futures_options}
Options on 3M SOFR futures are nowadays traded in the market. In our setup, the 3M futures rate $f^{{\rm 3M}}(t,S,T)$ admits an explicit representation as an exponentially affine function of $X_t$, as a consequence of \eqref{eq:future_3M}. Therefore, by applying the same arguments adopted in the proof of Theorem \ref{thm:forward}, one can derive a general pricing formula for options on 3M futures. In the context of a Gaussian HJM model (i.e., with deterministic volatility function), a valuation formula for options on 3M futures has been recently stated in \cite{Hen22}.
\end{rem}

For a 1M futures contract, the settlement is made at the rate $\log(B_T/B_S)/(T-S)$ (representing the arithmetic average of overnight rates over the period $[S,T]$, with $T-S$ being equal to one month). Denoting by $f^{{\rm 1M}}(t,S,T)$ the 1M futures rates at time $t$, it holds that
\be	\label{eq:future_1M}
f^{{\rm 1M}}(t,S,T) = \frac{1}{T-S}\,\EE\left[\int_S^Tr_u\ud u\bigg|\cF_t\right]
= \frac{1}{T-S}\bigl(L(S,T)+\EE[Y_T-Y_S|\cF_t]\bigr),
\ee
provided that the expectation is finite. In the statement of the following result, we denote by $\mu_i$, $i=0,1,\ldots,m$, the L\'evy measures appearing in \eqref{eq:LevyKhint} and by $\Phi_u$ and $\Phi_v$ the partial derivatives of  $\Phi$ with respect to its second and third argument, respectively, and similarly for $\Psi_u$ and $\Psi_v$.

\begin{prop}	\label{prop:futures_1M}
Suppose that the following condition holds:
\be	\label{eq:moment_condition}
\int_{D\setminus\{0\}}|\xi|\mu_i(\ud\xi)<+\infty,
\qquad\text{ for all }i=0,1,\ldots,m.
\ee
Then, for all $0\leq t\leq S<T<+\infty$, the 1M futures rate $f^{{\rm 1M}}(t,S,T)$ is given by
\be	\label{eq:futures_1M_expl}	\begin{aligned}
f^{{\rm 1M}}(t,S,T) &= \frac{L(S,T)+\Phi_v(T-t,0,0)-\Phi_v(S-t,0,0)}{T-S} \\
&\quad +\frac{\langle \Psi_v(T-t,0,0)-\Psi_v(S-t,0,0),X_t\rangle}{T-S}.
\end{aligned}	\ee
\end{prop}
\begin{proof}
Condition \ref{eq:moment_condition} implies that $\EE[|Y_T|]<+\infty$, for all $T>0$ (see \cite[Proposition 5.1]{FJR20}).
To compute the conditional expectation in the right-hand side of \eqref{eq:future_1M} we first compute the $\cF_t$-conditional characteristic function of $Y_T-Y_S$:
\begin{align*}
\EE\bigl[e^{\im\zeta(Y_T-Y_S)}\big|\cF_t\bigr]	
&= \EE\bigl[e^{\Phi(T-S,0,\im\zeta)+\langle\Psi(T-S,0,\im\zeta),X_S\rangle}\big|\cF_t\bigr]	\\
&= e^{\Phi(T-S,0,\im\zeta)+\Phi(S-t,\Psi(T-S,0,\im\zeta),0)+\langle\Psi(S-t,\Psi(T-S,0,\im\zeta),0),X_t\rangle},
\quad\text{ for }\zeta\in\R.
\end{align*}
Therefore, we have that
\begin{align*}
&\EE[Y_T-Y_S|\cF_t]
= -\im\frac{\ud}{\ud\zeta}\EE\bigl[e^{\im\zeta(Y_T-Y_S)}\big|\cF_t\bigr]\Big|_{\zeta=0}	\\
&= \Phi_v(T-S,0,0)+\Phi_u(S-t,0,0)\Psi_v(T-S,0,0)+\langle \Psi_u(S-t,0,0)\Psi_v(T-S,0,0),X_t\rangle,
\end{align*}
where the differentiability of the functions $\Phi$ and $\Psi$ is ensured by condition \eqref{eq:moment_condition} (see \cite[Lemmata 5.3 and 6.5]{dfs03}). Formula \eqref{eq:futures_1M_expl} then follows by relying on the semiflow relations \eqref{eq:semiflow}.
\end{proof}

Formula \eqref{eq:futures_1M_expl} requires the knowledge of the solution $(\Phi,\Psi)$ to the Riccati system \eqref{eq:Riccati1}-\eqref{eq:Riccati2}. There exists however an alternative procedure, which only relies on the functional characteristics $(F,R)$, explicitly given in \eqref{eq:LevyKhint}. The next proposition follows directly from \cite[Lemma 5.2]{FJR20}.

\begin{prop}	\label{prop:futures_1M_alt}
Suppose that condition \eqref{eq:moment_condition} holds. Then, for all $0\leq t\leq S<T<+\infty$, the 1M futures rate $f^{{\rm 1M}}(t,S,T)$ is given by
\be	\label{eq:futures_1M_expl_alt}
f^{{\rm 1M}}(t,S,T)
= \frac{1}{T-S}\left(L(S,T)+\int_S^T\bigl\langle\Lambda,\EE[X_u|\cF_t]\bigr\rangle\ud u\right),
\ee
where
\[
\EE[X_u|\cF_t] = e^{(u-t)\cA}X_t + \int_0^{u-t}e^{sA}\bar{b}\,\ud s,
\]
with the matrix $\cA\in\R^{d\times d}$ and the vector $\bar{b}\in\R^d$ given by
\[
\cA_{ij} := \partial_{u_i}R_j(u)\big|_{u=0}
\quad\text{and}\quad
\bar{b}_i := \partial_{u_i}F(u)\big|_{u=0},
\qquad\text{ for all }i,j=1,\ldots,m.
\]
\end{prop}

In the specific case of Gaussian affine diffusive models, the futures pricing formula stated in Proposition \ref{prop:futures_1M_alt} reduces to the explicit formula utilized in the recent work \cite{SS21}.

\appendix
\section{Fourier decomposition of payoff functions}

In this appendix, we recall the following well-known result (see, e.g., \cite[Lemma 10.2]{fil09}).

\begin{lem}	\label{lem:Fourier_payoffs}
Let $k>0$. For any $x\in\R_+$ the following holds:
\[
\frac{1}{2\pi}\int_{\R}x^{w+\im\lambda}\frac{k^{-(w-1+\im\lambda)}}{(w+\im\lambda)(w-1+\im\lambda)}\ud\lambda
= \begin{cases}
(k-x)^+,&\text{ if }w<0,\\
(x-k)^+-x = (k-x)^+-k,&\text{ if }0<w<1,\\
(x-k)^+,&\text{ if }w>1.
\end{cases}
\]
\end{lem}

\bibliographystyle{alpha}
\bibliography{biblio_RFR}

\end{document}